\documentclass[11pt]{article}
\usepackage{amsfonts,amssymb,amsmath,amsthm, fullpage}

\newcommand{\etalchar}[1]{$^{#1}$}

\DeclareMathOperator{\polylog}{polylog}
\DeclareMathOperator{\poly}{poly}

\newtheorem{theorem}{Theorem}

\newtheorem{lemma}{Lemma}

\newcommand{\eps}{\varepsilon}

\newcommand{\cV}{{\cal V}}

\newcommand{\bE}{\mathbb{E}}

\newcommand{\ket}[1]{|#1\rangle}

\newcommand{\QMA}{\mathsf{QMA}}
\newcommand{\NP}{\mathsf{NP}}
\newcommand{\LOCCQMA}{\mathsf{LOCCQMA}}
\newcommand{\BellQMA}{\mathsf{BellQMA}}

\begin{document}

\title{Short Multi-Prover Quantum Proofs for SAT \\
without Entangled Measurements}

\author{Jing Chen
\thanks{MIT.  Email: jingchen@csail.mit.edu.}
\and
Andrew Drucker
\thanks{MIT.  Email: adrucker@mit.edu.  Supported by a DARPA YFA grant.  Supported during part of this work by an Akamai Presidential Graduate Fellowship.}
}
\date{}
\maketitle

\begin{abstract}
BellQMA protocols are a subclass of multi-prover quantum Merlin-Arthur protocols in which the verifier is restricted to perform nonadaptive, unentangled measurements on the quantum states received from each Merlin.  In this paper, we prove that $m$-clause 3-SAT instances have BellQMA proofs of satisfiability with constant soundness gap, in which $\tilde{O}(\sqrt{m})$ Merlins each send $O(\log m)$ qubits to Arthur.  Our result answers a question of Aaronson et al., who gave a protocol with similar parameters that used entangled measurements; the analysis of our protocol is significantly simpler than that of Aaronson et al.  Our result also complements recent work of Brandao, Christandl, and Yard, who showed upper  bounds on the power of multi-prover quantum proofs with unentangled but adaptive (LOCC) measurements.
\end{abstract}

\section{Introduction}

In quantum Merlin-Arthur (QMA) proof systems, a computationally unbounded but untrusted prover Merlin tries to convince a polynomial-time quantum verifier Arthur that a given statement is true, by sending to Arthur a quantum state as a ``proof''.  We desire that the protocol have two properties.  The first is  ``completeness'': if the statement is true, then there should exist a proof which makes Arthur accept with at least some high probability $c$.  The second is ``soundness": if the statement is false, then for any proof received, Arthur should accept with at most some lower probability $s < c$. In general, the complexity class $\QMA_{\ell, c, s}$ consists of all languages whose membership can be proved by a quantum Merlin-Arthur proof system using $\ell$-qubit proofs, with completeness $c$ and soundness $s$. The complexity class $\QMA$ is defined to be $\QMA_{\mbox{\scriptsize{poly}}(n), 2/3, 1/3}$ where $n$ is the input length.

The generalized multi-prover version of $\QMA_{\ell, c, s}$, denoted $\QMA(k)_{\ell, c, s}$, was introduced by Kobayashi, Matsumoto, and Yamakami in \cite{KMY'03}. In such a proof system, $k$ Merlins are trying to convince a single Arthur that a given statement is true, by each sending Arthur a quantum state with $\ell$ qubits, and these $k$ states are assumed to be unentangled with each other.  The class $\QMA(k)$ is defined to be $\QMA(k)_{\mbox{\scriptsize{poly}}(n), 2/3, 1/3}$.

One piece of evidence for the power of multiple quantum provers was given by Blier and Tapp~\cite{BT'09}, who showed that every language in $\NP$ has a 2-prover proof system with extremely short proofs, of $\ell = O(\log n)$ qubits each.  Unfortunately, the soundness gap in their proof system (i.e., the quantity $c - s$) is very small: their protocol has $c = 1, s = 1 - 1/\poly(n)$.  A related but incomparable result was shown by Aaronson et al.~\cite{ABDFS'09}: they showed that $m$-clause 3-SAT instances can be proved satisfiable by a proof system using $\tilde{O}(\sqrt{m}))$ Merlins, each sending $O(\log m)$ bits, and with an improved soundness gap $c = 1, s = 1 - \Omega(1)$.  This still gives an almost-quadratic improvement in total proof length compared to known classical proofs, at least in the regime where the number $n$ of variables satisfies $n = \Theta(m)$.  

A recent paper by Harrow and Montanaro~\cite{HM'10} answers several important questions about $\QMA(k)$. They prove that soundness amplification for $\QMA(k)$ is possible and that $\QMA(k) = \QMA(2)$, for any $k = O(\poly(n))$.  Building on the result of~\cite{ABDFS'09}, Harrow and Montanaro also show that there exists a 2-prover proof system with proof length $\tilde{O}(\sqrt{m})$ for $m$-clause 3-SAT instances.

In all results mentioned above, Arthur uses the so-called \emph{swap test} measurement~\cite{BCWdW'01} as an important step in the protocol.  This is an efficient method to test whether two unentangled states are approximately equal.  
A natural question thus arises: how crucial is the swap test to the power of multi-prover quantum proof systems?  The swap test is an example of an \emph{entangled measurement}, in which the states may become entangled by the measurement process; so more generally, how crucial are entangled measurements to these proof systems?  To make such questions formal, \cite{ABDFS'09} defined the complexity classes $\LOCCQMA(k)$ and $\BellQMA(k)$.
The class $\LOCCQMA(k)$ consists of all languages whose membership can be proved by a $k$-prover proof system where Arthur is constrained to make unentangled measurements on the states provided by the Merlins, but is allowed to make these measurements adaptively based on the outcome of previous measurements.  $\BellQMA(k)$ is the subclass of $\LOCCQMA(k)$ in which we additionally require that no choice of measurement depends on the outcomes of other measurements.  (For more precise definitions of LOCCQMA and BellQMA protocols, see Sec.~\ref{belldefsec}.)  Brandao~\cite{B'08} showed that $\BellQMA(k) = \QMA$ for constant $k$.  Quite recently Brandao, Christandl, and Yard~\cite{BCY'10} made a breakthrough in the study of entanglement, and used their techniques to show that $\LOCCQMA(k) = \QMA$ for constant $k$. The situation for growing values of $k$ remains unclear.\footnote{On the one hand, the ideas of~\cite{HM'10} rely on the swap test and do not apply to $\LOCCQMA$ and $\BellQMA$.
On the other hand, both the proof for $\BellQMA(k)=\QMA$ and that for $\LOCCQMA(k) = \QMA$ for constant $k$ blow up the total length of the proofs to $n^{\exp(\Omega(k))}$, and thus cannot be used for $k = \omega(1)$.
We can at least say that, if $\QMA(2)=\QMA$, then all classes here collapse to $\QMA$.}

\paragraph{Our contribution.} In this paper, we exhibit a BellQMA proof system for 3-SAT, which essentially matches the parameters of the earlier protocol of~\cite{ABDFS'09}. Formally, we prove the following theorem:

\begin{theorem}\label{mainthm} There is a BellQMA proof system which, given a 3-SAT instance with $m$ clauses, uses $\tilde{O}(\sqrt{m})$ Merlins, each of which sends $O(\log m)$ qubits.  The proof system has completeness $1 - \exp\{-\Omega(\sqrt{m})\}$ and soundness $1 - \Omega(1)$.
\end{theorem}

Our result shows that entangled measurement is not necessary for short proofs of membership in 3-SAT, and thus answers a question raised in~\cite[Sec. 6.3]{ABDFS'09}.  Our system (just barely) loses perfect completeness as achieved by the protocol of~\cite{ABDFS'09}, but retains the constant soundness gap of that protocol.  The analysis of our protocol is also significantly simpler than that of~\cite{ABDFS'09}, which may be viewed as another contribution of this work.

Our protocol also complements a negative result from~\cite{BCY'10} (although our work was independent of theirs).
Corollary 5 of~\cite{BCY'10} implies that if there exists a 2-prover LOCCQMA protocol for 3-SAT with proof length $o(\sqrt{m})$ and with constant soundness gap,
then there exists a deterministic algorithm solving 3-SAT in subexponential time.\footnote{This is not quite made explicit in~\cite{BCY'10}, so we elaborate. \cite[Corollary 5]{BCY'10} gives an explicit construction of certain mappings called ``approximate disentanglers" (for LOCC measurements).  Following the use of such mappings as described in~\cite{ABDFS'09} (where they were defined), a 2-prover LOCCQMA proof system for 3-SAT of proof length $o(\sqrt{m})$ would imply a single-prover QMA protocol for 3-SAT of proof length $\ell = o(m)$.  The maximum acceptance probability of such a protocol can be approximately determined using semidefinite-programming solvers in time $\poly(2^\ell) = 2^{o(m)}$.}     This result seems to pose a significant barrier to achieving shorter proof length using unentangled measurements.  Our positive result nearly reaches this barrier, except for the fact that we use more than 2 provers.
  If our protocol could be converted to a 2-prover BellQMA or LOCCQMA protocol with similar proof length, then (under the plausible assumption that 3-SAT requires exponential time) we would obtain a nearly tight understanding of the power of these restricted quantum proof systems for 3-SAT (and of many other $\NP$ languages, via standard reductions).

\paragraph{Our Techniques.}The construction of our proof system, which we sketch next, adapts techniques used by Blier and Tapp in their proof system for 3-colorability from~\cite{BT'09}, and combines them with sampling and PCP ideas similar to those used by Aaronson et al.~\cite{ABDFS'09}.  In the Blier-Tapp protocol, Arthur receives two states $\ket{\Psi^1}, \ket{\Psi^2}$ of form
\[ \ket{\Psi^i} =  \sum_{v, c} \alpha^i_{v, c} \ket{v}\ket{c}      , \quad{}  i = 1, 2,\]
where $v \in \{0, 1, \ldots, n - 1\}$ indexes a vertex in a graph $G$ to be properly 3-colored and $c \in \{0, 1, 2\}$ is a color for $v$.  The protocol randomly performs one of three tests on $\ket{\Psi^1}, \ket{\Psi^2}$:
\begin{itemize}
\item An ``Equality Test" uses the swap test to check that $\ket{\Psi^1}, \ket{\Psi^2}$ are nearly equal.  This is the only entangled measurement.
\item A ``Uniformity Test" uses the quantum Fourier transform to check that each state has amplitudes which are almost uniformly spread over the $n$ vertices.
\item A ``Consistency Test" directly measures the vertex and color registers on each proof, rejecting if it sees two adjacent vertices with the same color or two differently-colored copies of the same vertex.
\end{itemize}
This proof is extremely succinct---only two states of $\log n$ qubits each---but as mentioned, its soundness gap is only inverse-polynomial in $n$.  Intuitively\footnote{The actual soundness gap shown in~\cite{BT'09} is $\Omega(n^{-6})$, even worse than our sketch would suggest; this was improved to $\Omega(n^{-3 - \eps})$ by Beigi in~\cite{Bei'10} by a modified protocol that still uses two $O(\log n)$-sized proofs, but loses perfect completeness.} this is because, if the Merlins send proofs uniformly spread over the vertices, each equipped with a coloring violating only one edge constraint, then the Consistency test can only succeed if the two vertices sampled come from this edge, which happens with probability $2/n^2$.

We modify this protocol as follows.  First, we ask for $O(\sqrt{n})$ proofs instead of 2. 
The ``birthday paradox'' then ensures that the Consistency Test will turn up pairs of equal vertices.  These will cause rejection unless almost all vertices are nearly-unanimous in their colorings across the supplied proofs.  With this added assurance, we simply omit the Equality Test.  Our modified Uniformity Test ensures that there are enough states in which the amplitudes are almost uniformly spread over the vertices, although to improve the soundness of this test, we are forced to sacrifice perfect completeness.

In the Consistency Test, we now also expect to sample pairs of vertices adjacent in $G$.  However, this will only lead to rejection with noticeable probability if the constraint problem is ``highly unsatisfiable", in the sense that every coloring violates an $\Omega(1)$ fraction of the edge constraints.  To ensure this, we apply the size-efficient PCP reduction of Dinur~\cite{D'07} to our original 3-coloring problem (or 3-SAT instance), which incurs only a polylogarithmic blowup in the instance size.  This completes the sketch of our protocol and the basic ideas of the analysis; the full proof of correctness is slightly more involved.

\paragraph{Open Problems.} Some questions raise from our result and those mentioned above. The most immediate one is whether the number of provers in our system can be further reduced without expanding the total proof length much. In particular, is there a 2-prover LOCCQMA proof for 3-SAT with length $\tilde{O}(\sqrt{m})$?  As we discussed earlier, if 3-SAT requires exponential time, this would be a nearly tight result in terms of proof length.   Whether general 2-prover protocols can achieve even shorter proofs of satisfiability remains an interesting question.

It also seems promising to see whether the entanglement theory ideas of~\cite{BCY'10} can be extended to give a fuller understanding of entanglement between more than 2 quantum states.  As just one benefit, this could yield new information about the power of $\LOCCQMA(k)$ and $\BellQMA(k)$ for superconstant $k$.

\section{Preliminaries}

We assume familiarity with (uniform) polynomial-time quantum algorithms.  Such algorithms are describable by a polynomial-size quantum circuit with polynomially many auxiliary qubits; the circuit is required to be constructible by a classical logarithmic-space algorithm.

\subsection{BellQMA and LOCCQMA protocols}\label{belldefsec}

We now more formally define the restricted multi-prover proof systems called BellQMA and LOCCQMA protocols.  The complexity classes $\BellQMA(k), \LOCCQMA(k)$ are defined in perfect analogy with $\QMA(k)$, using these restricted protocols.\footnote{In~\cite{BCY'10}, the notations $\QMA_{\mathsf{LO}}(k) = \BellQMA(k)$ and $\QMA_{\mathsf{LOCC}}(k) = \LOCCQMA(k)$ are used.}
In BellQMA protocols, Arthur performs a so-called ``Bell test" upon the quantum proofs; in LOCCQMA protocols, Arthur performs a test involving only ``local operations and classical communication" (LOCC) upon the proofs.  This motivates the terminology.

Our definition of BellQMA protocols is slightly broader than that given in~\cite{B'08}, and we discuss the difference below.
The more general class of LOCCQMA protocols will not be important in this paper, but we choose to provide a definition since previous discussions presume familiarity with the framework of LOCC tests (see~\cite{Ben'96, Nie'99}).

In a $k$-prover QMA protocol, the verifier (Arthur) receives a classical input $x \in \{0, 1\}^n$, as well as $k$ ``proof" states $\ket{\Psi_1}, \ldots, \ket{\Psi_k}$ from $k$ provers (Merlins).  The $k$ proofs are required to be unentangled.  Arthur performs some quantum-polynomial time test on the proofs, after which he either accepts or rejects.
We say that a QMA protocol for a language $L \subseteq \{0, 1\}^*$ possesses \emph{completeness} $c$ and \emph{soundness} $s < c$ if:
\begin{enumerate}
\item If $x \in L$, some collection $\ket{\Psi_1}, \ldots, \ket{\Psi_k}$ causes Arthur to accept with probability at least $c$;
\item If $x \notin L$, any collection $\ket{\Psi_1}, \ldots, \ket{\Psi_k}$ causes Arthur to accept with probability at most $s$.
\end{enumerate}

In a $k$-prover BellQMA protocol, we restrict the form of Arthur's test as follows: Arthur first performs a polynomial-time quantum computation on $x$ alone.  The workspace is then measured fully, yielding a description of measurements $M_1, \ldots, M_k$ described by polynomial-size quantum circuits; the $i$-th measurement, which may output more than one bit, is required to act locally on the $i$-th proof.  The measurements are then performed, and we let $y_i$ denote the output of the $i$-th measurement.  Finally, Arthur performs a quantum polynomial-time computation on $(x, y_1, \ldots, y_k)$ to decide whether to accept or reject.

Since the measurements $M_1, \ldots, M_k$ are fully determined by the intermediate measurement and act separately on the $k$ unentangled proof states, the proof states remain unentangled after the $M_i$ are performed.
A second observation about BellQMA protocols is that the identities of the measurements $M_1, \ldots, M_k$ can be random variables, and these random variables need not be independent.\footnote{In our protocol, the measurements will be chosen in a dependent fashion; however, it is not hard to modify our protocol to make these choices independent (for a fixed input $x$), with only a constant-factor increase in the number of provers.}

In~\cite{ABDFS'09}, the definition of BellQMA protocols was informal and slightly ambiguous.  In Brandao's thesis~\cite{B'08}, the definition of BellQMA protocols required the final computation on $(x, y_1, \ldots, y_k)$ to be performed by a \emph{classical} polynomial-time algorithm.  We feel that, since Arthur is allowed to use arbitrary polynomial-time quantum measurements $M_i$ on the $k$ proofs, it is natural to allow polynomial-time quantum computations in the final stage.  Indeed, Brandao's proof in~\cite{B'08} that $\BellQMA(k) = \BellQMA$ works equally well if this final computation is allowed to be quantum.  The BellQMA protocol that we give in this paper actually obeys Brandao's more restrictive definition.

In LOCCQMA protocols, Arthur is allowed to repeatedly and adaptively choose measurements to perform on the proofs.  However, these measurements are required to act locally on a single proof state, and they must be performed when Arthur's workspace is in a computational basis state.  This forces the proofs to remain unentangled throughout the computation.

Formally, $k$-prover LOCCQMA protocols can be defined as follows.  Arthur's verification algorithm consists of a polynomial number $p(n)$ of stages.  Each stage $t \leq p(n)$ has the following form:

\begin{enumerate}  
\item Arthur first performs a polynomial-time quantum computation acting on his workspace qubits alone.  Arthur's full workspace is then measured, yielding a tuple $(i_t, M_t, z_t)$.  Here $M_t$ describes a polynomial-time quantum measurement to be performed locally on the $i_t$-th proof, and $z_t$ is an auxiliary memory string.
\item $M_t$ is then performed, yielding an outcome $y_t$ of one or more bits.  Arthur then begins the $(t+1)$-st stage with his workspace initialized to the computational basis state $\ket{y_t, z_t}$.  
\end{enumerate}

Finally, Arthur accepts or rejects based upon the first bit of $y_{p(n)}$.  We remark that Arthur is allowed to measure individual proof states more than once.

Note that BellQMA protocols can be defined as LOCCQMA protocols in which all measurements to be performed on the $k$ proof states are determined in the first computation phase and described by the string $z_1$, then nonadaptively performed in the following phases.

\subsection{Dinur's PCP reduction}\label{dinursec}

The recent version of the PCP Theorem given by Dinur~\cite{D'07} is a reduction from the Boolean Satisfiability problem to a so-called \textit{constraint graph} problem, or 2-CSP.  A constraint graph is an undirected graph (possibly with self-loops) along with a set $\Sigma$ of ``colors". For each edge $e = (u, v) \in E$ the constraint graph has an associated constraint $R_e: \Sigma \times \Sigma \rightarrow \{0, 1\}$.  A coloring $\tau: V\rightarrow \Sigma$ \textit{satisfies} the constraint $R_e$ if $R_e(\tau(u), \tau(v)) = 1$.  We say that $G$ is satisfiable if there exists a mapping $\tau$ that satisfies all constraints.  We say that $G$ is $(1 - \eta)$-unsatisfiable if for all mappings $\tau: V\rightarrow \Sigma$, the fraction of constraints satisfied by $\tau$ is at most $(1 - \eta)$.

\begin{theorem} {\cite[Thm. 8.1 and its proof]{D'07}} There exists a reduction $T$ from 3-SAT instances to 2-CSP instances, with the following properties:

\begin{enumerate}

\item Completeness: If $\varphi$ is a satisfiable formula, $T(\varphi)$ is a satisfiable 2-CSP instance;

\item Soundness: There exists an absolute constant $\eta > 0$ such that if $\varphi$ is unatisfiable, $G = T(\varphi)$ is $(1 - \eta)$-unsatisfiable;

\item Size-Efficiency: If $\varphi$ has $m$ clauses, then $|V(G)| = n = O(m\cdot \polylog m)$ and also $|E(G)| =  O(m\cdot \polylog m)$;

\item Alphabet Size: $|\Sigma| = K = O(1)$;

\item Regularity: $G$ is $d$-regular (with self-loops), where $d = O(1)$.
\end{enumerate}
\end{theorem}

The last point is not quite explicit in the main statement of Dinur's result, but can be readily extracted from her proof: simply apply the ``preprocessing" transformation of \cite[Lemma 1.9]{D'07}  to the graph output by her main reduction.  Also, Dinur's main reduction takes as input a constraint graph, not a formula, but we can simply begin by transforming any 3-SAT instance of $m$ clauses into an equivalent instance of an $\NP$-hard 2-CSP such as 3-Colorability, yielding a constraint graph whose number of edges is $O(m)$.

In our protocol, Arthur first performs the above reduction, yielding a 2-CSP $G$ on $n = \tilde{O}(m)$ vertices that is either satisfiable or $(1 - \eta)$-unsatisfiable.  We now describe our BellQMA protocol for the  problem, starting directly from the constraint graph $G$.

\section{The BellQMA protocol}
Given a constraint graph $G$, let the proof states Arthur receives be denoted $\ket{\Psi_1},\dots,\ket{\Psi_{C\sqrt{n}}}$, with $C$ a constant (to be determined later). Each $\ket{\Psi_i}$ consists of a ``vertex'' register with base states $\ket{0},\dots,\ket{n-1}$ (describable by $\lceil \log n \rceil $ qubits) and a ``color'' register with base states $\ket{0},\dots,\ket{K-1}$ (describable with $\lceil \log K \rceil = O(1)$ qubits). Let $\mu \triangleq C\sqrt{n}/K$. The verifier's protocol is given below.

\begin{center}
{\bf Verifier $\cV$:}
\end{center}

\begin{itemize}
\item Flip a fair coin. If Heads, do the Uniformity Test; if Tails, do the Consistency Test.

\item \textbf{Uniformity Test:}
\begin{itemize}
\item[1.] For each $\ket{\Psi_i}$, perform a Fourier transform $F_K$ on the color register and measure it. \\
Let $Z=\{i: \mbox{the color register of } \ket{\Psi_i}\ \mbox{is measured 0}\}$. If $|Z|< \frac{99\mu}{100}$, \emph{reject}; otherwise continue.

\item[2.] For each $\ket{\Psi_i}$ such that the measurement in Step 1 gets 0, perform a Fourier transform $F_n$ on the vertex register and measure it. If there exits a $\ket{\Psi_i}$ such that the measurement doesn't get 0, \emph{reject}; otherwise \emph{accept}.

\end{itemize}

\item \textbf{Consistency Test:}
\begin{itemize}
\item[1.] For each $\ket{\Psi_i}$, measure it and denote the value in the two registers as $(v_i, c_i)$.

\item[2.] If there exists two states $\ket{\Psi_i}$ and $\ket{\Psi_j}$ such that  $e=(v_i,v_j)\in E$ but $R_e(c_i,c_j)=0$, \emph{reject}.  Also \emph{reject} if $v_i = v_j$ but $c_i \neq c_j$.  Otherwise, \emph{accept}.

\end{itemize}

\end{itemize}

Note that, since $n = \tilde{O}(m)$, we have $ \tilde{O}(m)$ proofs, each consisting of $\log n + O(1) = O(\log m)$ qubits, as needed.  The verifier is clearly polynomial-time and performs only nonadaptive, unentangled measurements, so it defines a valid BellQMA protocol.

\subsection{Completeness of our protocol}
In the rest of the paper, we use $\hat{i}$ to denote the square root of $-1$, and reserve the symbol $i$ as an index of states sent by the provers.  We first consider the case where the 3-SAT instance $\varphi$ is satisfiable, so that the constraint graph $G$ is also satisfiable.

\begin{lemma}
If $G$ is satisfiable, then there exist (unentangled) states $\ket{\Psi_1},\dots,\ket{\Psi_{C\sqrt{n}}}$ such that $\cV$ accepts with probability at least $1-\exp\left(-\mu/(2\cdot 10^4)\right) = 1 - \exp\left(-\Omega(\sqrt{m})\right)$.
\end{lemma}

\begin{proof}
Let $\ket{\Psi_i} = \ket{\Psi}\triangleq \frac{1}{\sqrt{n}}\sum_{v=0}^{n-1}\ket{v}\ket{\tau(v)}$ for all $i\leq C\sqrt{n}$, where $\tau$ is a coloring satisfying the constraint graph $G$. Since $\tau$ is satisfying, the Consistency Test will accept with probability 1. Below we analyze the probability that the Uniformity Test will accept if that test is performed.

Observe that a Fourier transform on the color register changes $\ket{\Psi}$ into
\begin{equation}\label{eq1}
(I_n\otimes F_K)\frac{1}{\sqrt{n}}\sum_{v=0}^{n-1}\ket{v}\ket{\tau(v)} = \frac{1}{\sqrt{n}}\sum_{v=0}^{n-1}\ket{v}\frac{1}{\sqrt{K}}\sum_{k=0}^{K-1}\exp \left( \frac{2\pi \hat{i} \tau(v) k}{K} \right)\ket{k}.
\end{equation}
Therefore for each $\ket{\Psi_i}$, the measurement in Step 1 of the Uniformity Test will see 0 with probability $n(1/\sqrt{n})^2(1/\sqrt{K})^2 = 1/K$. Accordingly, $\bE[|Z|] = C\sqrt{n}/K = \mu$. Since the $\ket{\Psi_i}$'s are unentangled, their measurement outcomes are independent. By a Chernoff bound, the probability that the Uniformity Test passes Step 1 is
\begin{align*}
1 - \Pr\left[|Z|< \frac{99\mu}{100}  \right]  >  1-\exp\left(-\frac{\mu}{2\cdot 10^4}\right).
\end{align*}
Further notice that according to Eq.~\ref{eq1}, conditioned on the color register measuring to 0 in Step 1 of the Uniformity Test, the state in the vertex register of $\ket{\Psi}$ becomes $\frac{1}{\sqrt{n}}\sum_{v=0}^{n-1}\ket{v}$, and a Fourier transform $F_n$ will change this state into
\[F_n \frac{1}{\sqrt{n}}\sum_{v=0}^{n-1}\ket{v} = \frac{1}{\sqrt{n}}\sum_{v=0}^{n-1}\frac{1}{\sqrt{n}}\sum_{u=0}^{n-1}\exp \left(\frac{2\pi \hat{i} vu}{n} \right)\ket{u} = \ket{0}.\]
Thus for each $\ket{\Psi_i}$ which is measured 0 in Step 1 of the Uniformity Test, Step 2 of this test will measure 0 with probability 1. Accordingly, if the Uniformity Test passes Step 1, it will accept in Step 2 with probability 1.

Thus the probability that $\cV$ accepts is at least
$$\frac{1}{2}\cdot \left(1-2\exp\left(-\frac{\mu}{2\cdot 10^4}\right) \right)\cdot 1 + \frac{1}{2}\cdot 1 = 1-\exp\left(-\frac{\mu}{2\cdot 10^4}\right).$$
\end{proof}

\section{Soundness of our protocol}
Now we consider the case where the 3-SAT instance $\varphi$ is unsatisfiable, so that the constraint graph $G$ is $(1 - \eta)$-unsatisfiable.  We show that for any sequence of proof states $\ket{\Psi_1}, \dots,\ket{\Psi_{C\sqrt{n}}}$, $\cV$ will reject with probability $\Omega(1)$.  The proof depends on three lemmas, corresponding to three cases that cover all possible sequences of states sent by the Merlins.

First, we can assume without loss of generality that the states Arthur receives are pure states, since by convexity some sequence of pure states maximizes Arthur's acceptance probability.  For each $i\in [C\sqrt{n}]$, we can express $\ket{\Psi_i}$ as
\[\ket{\Psi_i}=\sum_{v=0}^{n-1} \alpha^i_v \ket{v} \sum_{j=0}^{K-1}\beta^i_{v,j}\ket{j},\]
where $\sum_{v=0}^{n-1}|\alpha^i_v|^2 = 1$ for each $i$, and $\sum_{j=0}^{K-1}|\beta^i_{v,j}|^2 = 1$ for each $i, v$.

Let $p^i_0$ be the probability that the color register of $\ket{\Psi_i}$ is measured 0 after the Fourier transform in Step 1 of the Uniformity Test (conditioned on our performing that test).  Let
\[Z' \triangleq \left\{i: p^i_0\geq \frac{1}{4K}\right\}.\]
We claim:

\begin{lemma}\label{lem:fourier1}
If $|Z'|\leq \frac{\mu}{2}$, then Step 1 of the Uniformity Test rejects with probability $\Omega(1)$.
\end{lemma}

\begin{proof}
Let $Z_1= Z\cap Z'$ and $Z_2= Z\setminus Z'$.  We have $|Z_1|\leq |Z'|\leq \mu /2$, and $\Pr[i\in Z_2] < 1/(4K)$ independently for every $i\in [C\sqrt{n}]$ . Let $W$ be a random subset of $[C\sqrt{n}]$ such that $\Pr[i\in W] = 1/(4K)$ independently for every $i$.  Then $|W|$ stochastically dominates $|Z_2|$ and we have $\bE[|W|]= C\sqrt{n}/(4K) = \mu/4$.
The probability that Step 1 of the Uniformity Test accepts is
\begin{align*}
 \Pr\left[|Z|\geq \frac{99\mu}{100}\right]
&= \Pr\left[|Z_1|+|Z_2|\geq \frac{\mu}{2}+\frac{\mu}{4}+\frac{24\mu}{100}\right] \\
&\leq  \Pr\left[|Z_2|\geq \frac{\mu}{4}+\frac{24\mu}{100}\right] \\
&\leq  \Pr\left[|W|\geq \frac{\mu}{4}+\frac{24\mu}{100}\right] \\
&\leq  \exp\left( -\frac{24^2}{25^2\cdot 2}\cdot \frac{\mu}{4}\right) = o(1),
\end{align*}
where we used a Chernoff bound.
\end{proof}

Let $\eps <\eta/20$ 
be a constant (recall that $\eta$ is the soundness constant in Dinur's PCP reduction), and for each $i\in [C\sqrt{n}]$, let 
\[R_i \triangleq \{v: v\in V, |\alpha^i_v|^2<1/(8Kn)\}.\]
The next lemma considers the case where one of the sets $R_i$ is noticeably large.

\begin{lemma}\label{lem:fourier2}
Suppose there exists $i\in Z'$ such that $|R_i|\geq \eps n$.  Then the Uniformity Test rejects with probability $\Omega(1)$.
\end{lemma}

\begin{proof}
We focus on any such index $i$. After the Fourier transform on the color register in Step 1 of the Uniformity Test, $\ket{\Psi_i}$ becomes the state $\ket{\Phi_i}$ defined by
\begin{eqnarray}\label{eq2}
\ket{\Phi_i}& \triangleq & (I_n\otimes F_K)\sum_{v=0}^{n-1} \alpha^i_v \ket{v} \sum_{j=0}^{K-1}\beta^i_{v,j}\ket{j}
=\sum_{v=0}^{n-1}\alpha^i_v\ket{v}\sum_{j=0}^{K-1}\beta^i_{v,j}\frac{1}{\sqrt{K}}\sum_{k=0}^{K-1}\exp\left( \frac{2\pi \hat{i} jk}{K}\right)\ket{k}\nonumber \\
&=& \frac{1}{\sqrt{K}}\sum_{k=0}^{K-1}\left(\sum_{v=0}^{n-1}\alpha^i_v\left(\sum_{j=0}^{K-1}\beta^i_{v,j}\exp\left( \frac{2\pi\hat{i}jk}{K}\right)\right)\ket{v}\right)\ket{k}.
\end{eqnarray}
Let $\ket{ \gamma } = \sum_{v=0}^{n-1}\gamma^i_v\ket{v}$ with $\sum_{v=0}^{n-1}|\gamma^i_v|^2 = 1$ be the state left in the vertex register of $\ket{\Phi_i}$, after conditioning on the color register of $\ket{\Phi_i}$ measuring to 0, which occurs with probability $p^i_0$ by definition.
For each $v\in \{0, 1, \ldots, n-1\}$, let $P^i_{0, v}$ be the probability that the color register of $\ket{\Phi_i}$ is measured $0$ {\em and} the vertex register of $\ket{\Phi_i}$ is measured $v$. We have that
$$P^i_{0, v} = p^i_0 \cdot |\gamma^i_v|^2.$$
On the other hand, by Eq.~\ref{eq2} we have
$$P^i_{0, v} = \left|\frac{\alpha^i_v}{\sqrt{K}} \sum_{j=0}^{K-1}\beta^i_{v,j}\right|^2  =  \frac{|\alpha^i_v|^2}{K}\left|\sum_{j=0}^{K-1} \beta^i_{v,j}\right|^2  \leq  \frac{|\alpha^i_v|^2}{K} \cdot K\sum_{j=0}^{K-1}|\beta^i_{v,j}|^2  =  |\alpha^i_v|^2,$$
where we used the Cauchy-Schwarz inequality and the fact that $\sum_{j=0}^{K-1}|\beta^i_{v,j}|^2=1$.  Combining the above two equations, we find $p^i_0 \cdot |\gamma^i_v|^2 \leq |\alpha^i_v|^2$.
Because $i\in Z'$, we have $p^i_0 \geq 1/(4K)$. Thus $|\gamma^i_v|^2 \leq 4K|\alpha^i_v|^2$ for each $v$.
Accordingly, for each $v\in R_i$,
\[|\gamma^i_v|^2\leq \frac{4K}{8Kn} = \frac{1}{2n}.\]
Define
\[\ket{\psi} \triangleq F_n^{-1}\ket{0}  = \frac{1}{\sqrt{n}}\sum_{v = 0}^{n - 1}\ket{v}. \]
For each  $v \in R_i$, $| \gamma^i_v - 1/\sqrt{n} |^2 \geq  |1/\sqrt{2n}- 1/\sqrt{n}|^2  = (1 - 1/\sqrt{2})^2/n$.  Then, using unitarity of $F_n$, we have
\[      || F_n \ket{\gamma} -  \ket{0} ||_2^2  =  || \ket{\gamma} - \ket{\psi}  ||_2^2  \geq  \sum_{v \in R_i} (1 - 1/\sqrt{2})^2 /n \geq \eps n \cdot (1 - 1/\sqrt{2})^2 / n  = \Omega(1). \]
Since $||F_n \ket{\gamma}||_2 = 1$, it follows that the amplitude of $\ket{0}$ in $F_n \ket{\gamma}$ is of norm $\leq 1 - \Omega(1)$.  Thus if the color register of $\ket{\Phi^i}$ measures 0 in the Uniformity Test (as happens with probability $p^i_0 \geq 1/(4K)$ since $i \in Z'$), the vertex register measures to some $v \neq 0$ with probability $\Omega(1)$.  The Uniformity Test's rejection probability is therefore $\Omega(1/K) = \Omega(1)$ as claimed.
\end{proof}

In light of Lemmas~\ref{lem:fourier1} and~\ref{lem:fourier2}, we need only to address the case when $|Z'|>\mu/2$ and $|R_i|<\eps n$ for all $i\in Z'$.  We show that in this case, the Consistency Test rejects with probability $\Omega(1)$.

Consider an arbitrary state index $i \in Z'$.  Let $D_{i}$ denote the distribution on vertex/color pairs when $\ket{\Psi_{i}}$ is measured by the Consistency Test.  We can equivalently generate each $D_i$ as $D_i = g_{i}(U_{i})$, where each $U_{i}$ is a uniform, independent value from $[0, 1]$, and $g_{i}: [0, 1] \rightarrow V(G) \times \Sigma$ is a function such that each preimage $g_{i}^{-1}((v, c))$ is an interval of length equal to $\Pr[D_{i} = (v, c)]$.
 Then for each $v \notin R_i$, the set $g_{i}^{-1}(v, \star)$ is of measure $|\alpha^i_v|^2\geq 1/(8Kn)$. Select $J_{i, v} \subseteq g_{i}^{-1}(v, \star)$ of measure \emph{exactly} $1/(8Kn)$ for each such $v$, and let $J_i = \bigcup_{v \notin R_i} J_{i, v}$.  Observe the following: first, $J_i$ has measure greater than $(1 - \eps)/(8K)$.  Second, conditioned on $U_i \in J_i$, the posterior distribution of the vertex $v_i$ that $g_i$ outputs is now \textit{uniform} over $S_i\triangleq \{0, 1, \ldots, n-1\}\setminus R_i$.

So let us consider the Consistency Test applied to a sequence of states satisfying $|Z'|>\mu/2$ and $|R_i|<\eps n$ for all $i\in Z'$.  Letting the measurement outcomes be generated as described above, define the random set
\[Z'' \triangleq \{i: i\in Z', U_i\in J_i\} .\]
Notice that $Z''$ is itself a random variable determined by the $U_i$'s. Notice also that for each $i\in Z'$, the probability that $i\in Z''$ is at least $(1-\epsilon)/(8K)$, and these events are independent from each other.
Therefore we have that
\[\bE[|Z''|] \geq \frac{(1-\epsilon)|Z'|}{8K} > \frac{(1 - \eps)\mu}{16K} = \frac{(1 - \eps)C\sqrt{n}}{16K^2}.\]
Since $|Z''|$ never exceeds $C\sqrt{n}$, the total number of proof states, we find that with probability $\Omega(1)$,
\begin{equation}\label{eq3}
|Z''| \geq  \frac{C\sqrt{n}}{32K^2} .
\end{equation}
The following lemma tells us that if $C$ is chosen as a suitably large constant, then conditioned on Eq.~\ref{eq3} holding, the Consistency Test rejects with $\Omega(1)$ probability.

\begin{lemma}\label{lem:fourier3}  Let $(G, \{R_{e}\})$ be an $n$-vertex, $d$-regular constraint graph (possibly with self-loops, and $d > 1$) with alphabet $K$, such that $G$ is $(1 - \eta)$-unsatisfiable.  Let $D_1, \ldots D_{m'}$ be independent distributions on $V(G) \times \Sigma$, with $(v_i, c_i)$ denoting the output of $D_i$.  Suppose for each $i \leq m'$ there exists an $S_i \subseteq V(G)$ of size at least $(1 - \eps)n$, such that $v_i$ is uniformly distributed over $S_i$, where $\eps < \eta/20$.

Then we can set $m' = O(\sqrt{n})$ large enough so that  with probability at least $.99$ there exists $i < j \leq m'$ such that: \textit{either} $e = (v_i, v_j)$ is an edge of $G$ and $R_e(c_i, c_j) = 0$; \textit{or} $v_i = v_j, c_i \neq c_j$.
(The constant in the $O()$ notation depends on $d$ and $\eta$, but not $K$.)
\end{lemma}

To apply Lemma~\ref{lem:fourier3} to our Consistency Test when $|Z'|>\mu/2$ and $|R_i|<\eps n$ for all $i\in Z'$, choose $C = O(1)$ such that $C\sqrt{n}/(32K^2) \geq m'$. Then conditioned on $|Z''|\geq m'$, which occurs with probability $\Omega(1)$, we can select $D_1,\dots, D_{m'}$ from the distributions of $\ket{\Psi_i}$'s (when measured in the Consistent Test) such that $i\in Z''$. By definition of $Z''$, these $D_1,\dots, D_{m'}$ satisfy the hypothesis of Lemma~\ref{lem:fourier3}, so the Consistency Test rejects with probability .99.   Thus in this case $\cV$ also rejects with probability $\Omega(1)$.  This completes the proof that $\cV$ possesses soundness $1-\Omega(1)$, proving Theorem~\ref{mainthm}.

\begin{proof}[Proof of Lemma~\ref{lem:fourier3}] For $i < j \leq m'$, let $V_{i,j}$ be the indicator random variable for the event that \textit{either} $e = (v_i, v_j) \in E(G)$ and $R_e(c_i, c_j) = 0$, \textit{or} $v_i = v_j$ and $c_i \neq c_j$.  Let $V = \sum_{i < j} V_{i, j}$.  To prove Lemma \ref{lem:fourier3} it is enough to show that Pr$[V = 0] \leq .01$.  We show this using the second moment method.

A first observation is that we can generate $D_i$ in the following way: first, randomly select a coloring $\tau_i$ according to some distribution $H_i$; next select $v_i$ uniformly from $S_i$, and set $c_i = \tau_i(v_i)$.  To be explicit, each $H_i$ independently chooses colors according to the rule Pr$[\tau_i(v) = c] = $ Pr$[c_i = c| v_i = v]$.  It is easily verified that this process yields $D_i$.

Next we lower-bound $\bE[V] = \sum_{i < j}\bE[V_{i, j}]$.  Fix any pair $i, j$, $1 \leq i < j \leq m'$.  Condition on any values of the colorings $\tau_i, \tau_j$; we'll show that $\bE[V_{i,j}|\tau_i, \tau_j] \geq  \eps/ n$.   Let $P_{i, j} \subseteq V(G)$ be the subset of vertices $v$ for which $\tau_{i}(v) = \tau_j(v)$.  Suppose first that $|P_{i, j}| \leq (1 - 3\eps) n$.  In this case there are at least $\eps n$ vertices contained in $S_i \cap S_j \cap \overline{P_{i, j}}$, and $V_{i, j} = 1$ whenever a vertex in this set is selected as both $v_i$ and $v_j$.  Thus in this case $\bE[V_{i,j}|\tau_i, \tau_j] \geq \eps n\cdot |S_i|^{-1} \cdot |S_j|^{-1} \geq  \eps/n$.

For our second case, suppose $|P_{i, j}| > (1 - 3\eps) n$.  Consider the induced subgraph $G[S_i \cap S_j \cap P_{i, j}]$, which contains at least $n - 2\eps n - 3\eps n = (1 - 5\eps)n$ vertices.  Since $G$ has maximum degree $d$, $|E(G)| = dn/2$, and the set $\overline{S_i} \cup \overline{S_j} \cup \overline{P_{i, j}}$ is incident on at most $d(5\eps n)$ edges, we have that $|E(G[S_i \cap S_j  \cap P_{i, j}])| \geq dn/2 - 5d \eps n = (1 - 10\eps)(dn)/2 = (1 - 10\eps)|E(G)|$.  By $(1 - \eta)$-unsatisfiability of $(G, \{R_e\})$, the coloring $\tau_i$ satisfies at most a $(1 - \eta)/(1 - 10\eps)$ fraction of the edge constraints in $G[S_i \cap S_j \cap P_{i, j}]$.  Thus the fraction of these constraints which are \textit{violated} by $\tau_i$ is at least $1 - \frac{1 - \eta}{1 - 10\eps} = \frac{1 - 10\eps -( 1- \eta )}{1 - 10\eps} > \frac{\eta}{2(1 - 10\eps)} > \frac{\eta}{2}$, since $\eta > 20\eps$.

Now $\tau_j \equiv \tau_i$ on $P_{i, j}$.  We can thus lower-bound $\bE[V_{i, j}|\tau_i,\tau_j]$ by the probability that $v_i, v_j \in S_i \cap S_j \cap P_{i, j}$ and that $v_i, v_j$ form (in either order) an edge violated by the color assignment $(\tau_i(v_i), \tau_j(v_j))$.  Note that some edges are self-loops and so may only be chosen in one way.  We get
\[  \bE[V_{i, j}|\tau_i, \tau_j]   \geq (1 -  5\eps)^2\cdot  \frac{ \frac{\eta}{2}|E(G[S_i \cap S_j \cap P_{i, j}])|}{n^2} \]
\[  \geq  (1 -  5\eps)^2\cdot  \frac{ \eta (1 - 10\eps)|E(G)|    }{2n^2}    =  \frac{\eta(1 -  5\eps)^2(1 - 10\eps)d}{4n} .  \]
Recall that $\eps < \eta/20 < 1/20$, so the quantity above is greater than $\eta\cdot 2^{-3}d/4n > \eta/(20n) > \eps/n$, as needed (using $d > 1$).  Thus in either of our two cases we conclude $\bE[V_{i, j}|\tau_i, \tau_j] \geq \eps/n$, so $\bE[V_{i, j}] \geq \eps/n$ unconditioned as well.  Summing over all $i < j$, we find $\bE[V] \geq \eps {m' \choose 2}/n = \Omega\left((m')^2/n \right)$.

\vspace{.5 em}

Next we upper-bound $\bE[V^2] = \sum_{i < j, k < l} \bE[V_{i, j}V_{k, l}]$.  There are ${m' \choose 2}$ terms for which $(i, j) = (k, l)$.  For each such term $\bE[V_{i,j}^2] = \bE[V_{i, j}]$.  Condition on the vertex $v_i$ outputted by $D_i$.  Fixing any such choice of $v_i$, the probability that $V_{i, j} = 1$ is of course upper-bounded by the probability that $v_j$ is equal or adjacent to $v_i$ in $G$.  Since $v_j$ is uniform on $S_j$ and $v_i$ is of degree $d$, this probability is at most $(d+1)/|S_i| \leq (d + 1)/((1 - \eps)n)$, so $\bE[ V_{i, j} | v_i] \leq (d+1)/((1 - \eps)n)$.  As $v_i$ was an arbitrary conditioning, we conclude $\bE[ V_{i, j}] \leq (d + 1)/((1 - \eps)n)$.  Thus the contribution to $\bE[V^2]$ from terms where $(i, j) = (k, l)$ is at most ${m' \choose 2}(d + 1)/((1 - \eps)n) = O\left((m')^2/n\right)$.

If $(i, j), (k, l)$ consists of three distinct indices, assume that $j = l$, the other cases being handled similarly.  Condition on any choice of $v_j$.  Then $V_{i,j}V_{j, k} = 1$ can only occur if $v_j$ and $v_k$ are each either incident on or equal to $v_i$.  These two events are independent after conditioning on $v_j$ since $D_i, D_j, D_k$ are independent.  Thus $\bE[V_{i, j}V_{k, j}] \leq \left[(d + 1)/((1 - \eps)n)\right]^2$.

For any three distinct indices $a < b < c \leq m'$, there are six tuples $(i < j), (k < l)$ for which $\{i, j, k, l\} = \{a, b, c\}$. Thus the contribution to $\bE[V^2]$ from these ``triplet'' terms is at most $6{m' \choose 3} \cdot \left[(d + 1)/((1 - \eps)n)\right]^2 = O\left( (m')^3/n^2 \right)$.

If $(i, j), (k, l)$ are four distinct elements of $[m']$, then the pair $V_{i, j}, V_{j, k}$ depend on disjoint sets of independent random variables, so that $\bE[V_{i, j}V_{k, l}] = \bE[V_{i, j}]\bE[V_{k, l}]$.  Thus the contribution to $\bE[V^2]$ from these terms is upper-bounded by $\sum_{i < j, k < l}  \bE[V_{i, j}]\bE[V_{k, l}] = \bE[V]^2$.

Putting things together,
\[\bE[V^2] <  O\left(  \frac{(m')^2}{n}  + \frac{(m')^3}{n^2} \right)  +  \bE[V]^2     . \]
With this bound in hand, we apply Chebyshev's inequality:
\begin{align*}
\Pr[V = 0] &\leq \Pr\left[|V - \bE[V]| \geq \bE[V] \right]\\   &\leq \frac{\bE[V^2] - \bE[V]^2}{\bE[V]^2} \\
&\leq    O \left( \frac{n^2}{(m')^4} \left(  \frac{(m')^2}{n}  + \frac{(m')^3}{n^2} \right)  \right)        \\
&=   O\left( \frac{n}{(m')^2}  + \frac{1}{m'}      \right) ,
\end{align*}
which is at most $.01$ if we take $m'$ to be a suitably large value in $O(\sqrt{n})$.  This proves the lemma.
\end{proof}

\section{Acknowledgements}
We thank Scott Aaronson for helpful comments and suggestions.

\bibliographystyle{halpha}

\end{document}